\documentclass[10pt,journal]{IEEEtran}

\usepackage[utf8]{inputenc} 
\usepackage[T1]{fontenc}
\usepackage{url}
\usepackage{ifthen}
\usepackage{cite}
\usepackage[cmex10]{amsmath} 
\usepackage{graphicx,colordvi,psfrag}
\usepackage{amsmath,amssymb}
\usepackage{epstopdf}
\usepackage[caption=false]{subfig}
\usepackage{epsfig,cite}
\usepackage{calc, pgf, xcolor}
\usepackage{nicefrac}
\usepackage{enumerate}
\usepackage{bm}
\usepackage{fixltx2e}
\usepackage{pstricks}
\usepackage{url}
\usepackage{setspace}
\usepackage{dsfont}

\newtheorem{theorem}{Theorem}

\newtheorem{definition}{Definition}

\newtheorem{proposition}{Proposition}
\newenvironment{proof}[1][Proof]{\noindent\textbf{#1.} }{\ \rule{0.5em}{0.5em}}

\newcommand{\ZZ}{\mathbb{Z}}
\newcommand{\CC}{\mathbb{C}}
\newcommand{\RR}{\mathbb{R}}

\newcommand{\cC}{\mathcal{C}}

\newcommand{\cF}{\mathcal{F}}

\newcommand{\norm}[1]{\left|\left| #1 \right|\right|}
\newcommand{\abs}[1]{\left \vert #1 \right \vert}

\newrgbcolor{BoxRed}{1 .5 .5}
\newrgbcolor{BoxBlue}{.9 .9 1}
\newrgbcolor{lightgrey}{.7 .7 0.7}

\begin{document}

\title{Above the Nyquist Rate, Modulo Folding Does Not Hurt}
\author{Elad Romanov and Or Ordentlich
	\thanks{}
	\thanks{Elad Romanov and Or Ordentlich are with the School of Computer Science and Engineering, Hebrew University of Jerusalem, Israel (emails: \{elad.romanov,or.ordentlich\}@mail.huji.ac.il). }
	\thanks{This work was supported, in part, by ISF under Grant 1791/17, and by the GENESIS consortium via the Israel Ministry of Economy and Industry.
		ER acknowledges support from ISF grant 1523/16.}
	\thanks{}}

\maketitle

\begin{abstract}
We consider the problem of recovering a continuous-time bandlimited signal from the discrete-time signal obtained from sampling it every $T_s$ seconds and reducing the result modulo $\Delta$, for some $\Delta>0$. For $\Delta=\infty$ the celebrated Shannon-Nyquist sampling theorem guarantees that perfect recovery is possible provided that the sampling rate $1/T_s$ exceeds the so-called Nyquist rate. Recent work by Bhandari et al. has shown that for any $\Delta>0$ perfect reconstruction is still possible if the sampling rate exceeds the Nyquist rate by a factor of $\pi e$. In this letter we improve upon this result and show that for finite energy signals, perfect recovery is possible for any $\Delta>0$ and any sampling rate above the Nyquist rate. Thus, modulo folding does not degrade the signal, provided that the sampling rate exceeds the Nyquist rate. This claim is proved by establishing a connection between the recovery problem of  a discrete-time signal from its modulo reduced version and the problem of predicting the next sample of a discrete-time signal from its past, and leveraging the fact that for a bandlimited signal the prediction error can be made arbitrarily small.
\end{abstract}

\section{Introduction}
\label{sec:intro}

The Shannon-Nyquist sampling theorem guarantees that any signal $x(t)$ whose Fourier transform is supported on $[-W,W]$ can be perfectly reconstructed from the discrete-time signal $\{x_n=x(nT_s)\}$ as long as $T_s<1/2W$. However, in order to be digitally stored/processed, the signal $\{x_n\}$ must be further digitized, which is achieved by quantizing it to a set of finite cardinality.

In practice, quantization of the process $\{x_n\}$ is invariably done via a scalar uniform quantizer. For a signal with bounded dynamic range, say $x_n\in(-2^{m},2^m)$ for all $n\in\ZZ$, the operation of an $R$-bit scalar uniform quantizer can be roughly described as that of writing the binary expansion of each sample as $x_n=-a_{n,0}2^m+\sum_{i=1}^{\infty}a_{n,i}2^{m-i}$, and then discarding all but the $R$ most significant bits (MSBs), i.e., $x_n$ is represented by $a_{n,0},a_{n,1},\ldots,a_{n,R-1}$. Clearly, after this form of quantization, as well as any other finite bit-rate form of quantization, $x(t)$ can no longer be perfectly reconstructed.

Let $\{x^Q_n\}$ be the signal obtained by quantizing $\{x_n\}$ using a $R$-bit uniform quantizer, and $\{x^{\text{res}}_n=x_n-x^Q_n\}$ be the residual signal. Note that the signal $\{x^{\text{res}}_n\}$ corresponds to discarding the $R$ MSBs of each sample and keeping all the remaining least significant bits (LSBs). If one has to choose between reconstructing $\{x_n\}$ from either $\{x^Q_n\}$ or $\{x^{\text{res}}_n\}$, the intuitive choice will be $\{x^Q_n\}$. After all, $|x_n-x^Q_n|<2^m\cdot 2^{-(R-1)}$, whereas $|x_n-x^{\text{res}}_n|$ can be as large as $2^m(1-2^{-R})$. However, somewhat surprisingly, this letter shows that the opposite is true, provided that $T_s<1/2W$. In particular, we improve upon the recent results of Bhandari et al.~\cite{bhandari2017unlimited} and show that whenever the sampling frequency exceeds the Shannon-Nyquist frequency, for any $R\in\mathbb{N}$, the signal $\{x_n\}$ can be perfectly reconstructed from $\{x^{\text{res}}_n\}$.

To be more concrete, the setup we consider is that of sampling a modulo-reduced signal. For a given $\Delta>0$, and a real number $x\in\RR$, we define $x^*=[x]\bmod\Delta$ as the unique number in $[-\frac{\Delta}{2},\frac{\Delta}{2})$ such that $x-x^*\in\Delta\ZZ$. For a complex number $x=a+ib\in\CC$, the modulo operation corresponds to reducing both the real and the imaginary parts modulo $\Delta$, i.e., $x^*=a^*+i b^*$. 
The signal $x^*(t)$ is obtained by modulo reducing $x(t)$ at all times, and the discrete-time signal $\{x^*_n=x^*(nT_s)\}$ is obtained by sampling $x^*(t)$ every $T_s$ seconds. Note that $x^*_n=\left[x_n\right]^*$ for all $n\in\ZZ$, where $\{x_n=x(nT_s)\}$. The question we address in this letter is \emph{under what conditions on $T_s$ and $\Delta$ is it possible to exactly recover $x(t)$ from $\{x^*_n\}$?}

Our main result is that for all finite energy signals whose Fourier transform is supported on $[-W,W]$, and whose amplitude vanishes for $|t|$ large enough, perfect reconstruction is possible for any $\Delta>0$ and any $T_s<\frac{1}{2W}$. See Section~\ref{sec:construction} for the formal statement. 

\subsection{Related  Work}

In~\cite{bhandari2017unlimited}, it was shown that under essentially the same assumptions as above on the class of signals, perfect reconstruction is possible for any $\Delta>0$ and any $T_s<\frac{1}{\pi e}\frac{1}{2W}$. Thus, our results improve the sampling rate obtained in~\cite{bhandari2017unlimited} by a factor of $\pi e$ samples per seconds. Clearly, the sampling rate $T_s<\frac{1}{2W}$ found here is the best possible, since even without modulo reduction (alternatively, $\Delta=\infty$) this condition is required in order to guarantee perfect recovery within the considered class of signals.

The results of Bhandari et al.~\cite{bhandari2017unlimited} had already sparked a lot of follow up work, see e.g.~\cite{bkr18,ct17,mjg18,rass18,rbwp18,sh18,jt19}. A closely related recent line of work is that of modulo ADCs~\cite{othsw18,boufounos12,oe13b,ro19}, which are finite bit-rate ADCs that reduce their input signal modulo $\Delta$ prior to quantization. The main difference between the so-called \emph{unlimited sampling} framework of~\cite{bhandari2017unlimited} and the \emph{modulo ADC} framework is that the former does not consider the effect of quantization noise added to the modulo reduced signal, whereas the latter explicitly studies the tradeoff between quantization rate and distortion after modulo reduction. 

The setup considered in this letter falls within the framework of unlimited sampling. Namely, we are assuming $\{x^*_n\}$ is available at perfect precision, and we are only interested in characterizing the optimal tradeoff between $T_s$ and $\Delta$ for which $\{x_n\}$ (and consequently also $x(t)$) can be perfectly reconstructed despite the modulo folding. Our results indicate that modulo reduction per se, incurs no loss on our ability to reconstruct $x(t)$, provided that the Shannon-Nyquist condition is satisfied. Therefore, the main question of interest is what can be gained from modulo folding in terms of quantization rate. The results in~\cite{othsw18} show that for stationary Gaussian processes, oversampled modulo ADCs achieve rate-distortion tradeoff which is close to the information theoretic limits. Similar results were also obtained in~\cite{othsw18} for spatially correlated processes. 

Despite the differences between the unlimited sampling and modulo ADCs frameworks, the insights gained from developing recovery algorithms for oversampled modulo ADCs~\cite{othsw18} turn out to also be suitable for recovery under the unlimited sampling framework. In particular, the recovery algorithm outlined below is inspired by that of~\cite[Section 3]{othsw18}.

\subsection{Recovery Through Prediction}

Our main result may look somewhat absurd at first, as it seems that as we drive $\Delta$ to $0$, we should end up with no information left. However, the result becomes very intuitive once a connection is established between the problem of modulo unwrapping of a discrete-time signal and the problem of predicting the next sample of a discrete-time signal from its past. 

In particular, our algorithm for recovering $\{x_n\}$ from $\{x^*_n\}$ is of a sequential nature. Since we assume $|x(t)|$ is small for all large enough $|t|$, there exists some negative integer $N\in\ZZ$ such that $|x_n|<\frac{\Delta}{2}$ for all $n<N$, which implies that $x^*_n=x_n$ for all $n<N$. Next, we would like to recover $x_n$, for $n=N$, which is no longer guaranteed to satisfy $x^*_n=x_n$. To this end, we first compute a linear predictor $\hat{x}_n=\sum_{i=1}^{\infty} h_i x_{n-i}$ for $x_n$ from the past samples, that are available to us unfolded. Then, we compute 
\begin{align}
e^*_n=[x^*_n-\hat{x}_n]^*=[x_n-\hat{x}_n]^*,\nonumber
\end{align}
where $e_n:= x_n-\hat{x}_n$, and we have used the fact that the modulo operation is invariant to translation by integer multiples of $\Delta$, such that $[a^*+b]^*=[a+b]^*$ for any $a,b\in\CC$. Note that $e^*_n=e_n$ if $|e_n|<\frac{\Delta}{2}$, and in this case we can recover $x_n$ as $x_n=\hat{x}_n+e^*_n$. Thus, recovery of $x_n$ is possible if the prediction error $e_n=x_n-\hat{x}_n$ from its past is guaranteed to  have magnitude smaller than $\frac{\Delta}{2}$. Once we have recovered $x_n$ correctly, we can use it for computing the predictor $\hat{x}_{n+1}$ of $x_{n+1}$ from its past, and repeat the same procedure. See Figure~\ref{fig:recalgfig}. Our procedure therefore succeeds in recovering $\{x_n\}$ from $\{x^*_n\}$ provided that the prediction error process $\{e_n=x_n-\hat{x}_n\}$ is bounded in magnitude by $\Delta/2$ for all $n\in\ZZ$.

\begin{figure}[t]
\begin{center}
\psset{unit=0.6mm}
\begin{pspicture}(0,0)(250,40)

\rput(0,30){$x_n$}\psline{->}(5,30)(15,30)
\psframe(15,25)(42,35)\rput(28,30){$\bmod\Delta$}
\psline{->}(42,30)(53,30)

\psframe[linestyle=dashed](8,20)(45,40)\rput(18,43){Sampler}

\rput(-27,0){
\pscircle(85,30){5}\rput(85,30){$\Sigma$}\rput(78,27){$-$}
\psline{->}(90,30)(97,30)\psframe(97,25)(115,35)\rput(106,30){$\bmod\Delta$}
\psline{->}(115,30)(122,30)\pscircle(127,30){5}\rput(127,30){$\Sigma$}
\psline{->}(132,30)(150,30)
\rput(160,30){$x_n^{\text{recovered}}$}
\psline{->}(140,30)(140,15)\psframe(132,5)(148,15)\rput(140,10){Filter}
\psline{->}(132,10)(85,10)(85,25)
\psline{->}(132,10)(127,10)(127,25)
}

\psframe[linestyle=dashed](48,0)(144,40)\rput(63,43){Decoder}

\end{pspicture}
\end{center}
\caption{Schematic architecture for the proposed recovery algorithm.}
\label{fig:recalgfig}
\end{figure}
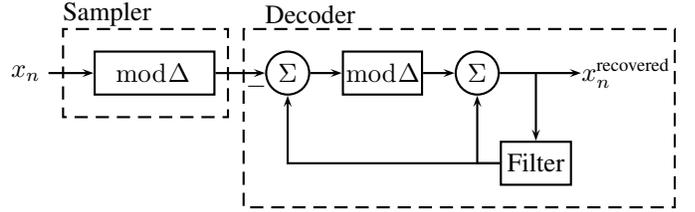 

Now, to establish our main result, it remains to show that $|e_n|$ can be made arbitrarily small at all times. This follows since the discrete-time process ${x_n}$ is bandlimited, i.e., its discrete-time Fourier transform (DTFT) is supported on $\left[-\gamma,\gamma\right]$, where $\gamma=W\cdot T_s<\frac{1}{2}$. To see this, note that 
\begin{align}
e_n=x_n-\sum_{i=1}^\infty h_i x_{n-i}=x_n\star c_n,\nonumber
\end{align}
where $c_n=\delta_n-h_n$ is a monic filter, i.e., a causal filter with first tap equal to $1$. Since the coefficients $h_1,h_2,\ldots$ are not constrained, we are free to take $\{c_n\}$ as any monic filter. In particular, for any $\epsilon>0$, we may choose $\{c_n\}$ to satisfy
\begin{align}
|C(f)|^2=\left|\sum_{k=1}^{\infty}c_k e^{-i2\pi f k}\right|^2=\begin{cases}
\epsilon & |f|\leq\gamma\\
\left(\frac{1}{\epsilon}\right)^{\frac{2\gamma}{1-2\gamma}} & |f|\in(\gamma,\frac{1}{2}]
\end{cases}.\nonumber
\end{align}
The existence of a monic filter with this frequency response is guaranteed by Paley-Wiener's spectral factorization theorem~\cite{gershogray}, and the fact that $\int_{-1/2}^{1/2}\log |C(f)|^2 df=0$. Now, since $\{x_n\}$ has no spectral energy in $|f|>\gamma$, we have that
\begin{align}
\|\{e_n\}\|^2_2=\|\{x_n\star c_n\}\|_2^2=\epsilon\|\{x_n\}\|^2.
\end{align}
Thus, if we take $\epsilon< \frac{\Delta^2}{4\|\{x_n\}\|^2}$ we get that $\|\{e_n\}\|_2<\frac{\Delta}{2}$, and in particular $|e_n|<\frac{\Delta}{2}$ for all $n\in\ZZ$. The argument is made precise in Section~\ref{sec:details}, where the infinite length filter used above is replaced with a finite length one, whose coefficients are designed based on Chebyshev's polynomials. 

We note that the recovery algorithm proposed in~\cite{bhandari2017unlimited} relied on $L$th-order discrete differentiation of $\{x^*_n\}$ reduced modulo $\Delta$, which is equivalent to $L$th-order discrete differentiation of $\{x_n\}$ reduced modulo $\Delta$, followed by $L$-th order discrete integration. In a way, this is equivalent to taking the filter $\{c_n\}$ above as $C(Z)=(1-Z^{-1})^L$ in the $Z$-domain. If the DTFT of $\{x_n\}$ is supported on a  frequency interval $[-f_0,f_0]$ for $0<f_0<1/2$ small enough, the signal $\{c_n\star x_n\}$ will never exceed $\Delta/2$ even with this choice of $\{c_n\}$, provided that $L$ is large enough. However, such choice of $\{c_n\}$ does not guarantee correct reconstruction for $f_0$ arbitrarily close to $1/2$, which is the reason for the loss of a constant factor in the required sampling rate reported in~\cite{bhandari2017unlimited}.

\section{The recovery scheme}
\label{sec:details}

\subsection{Setting and Notation}

Throughout, we will use the following convention for the Fourier transform: for a signal $x\in L^2(\RR)\cap L^1(\RR)$, its Fourier transform is given by
\[
\cF(x)(f) = \int_{\RR} x(t)e^{-2\pi i ft}dt\,,
\] 
with the inverse transform given by
\[
\cF^{-1}(\hat{x})(t) = \int_{\RR} \hat{x}(f)e^{2\pi i ft}df \,,
\]
so that the extension $\cF : L^2(\RR) \to L^2(\RR)$ is an isometry in the sense of Hilbert spaces.

\begin{definition}
For $W,E,T_0,\rho \in\RR_+$, the class of signals $\cC_{W,E,T_0,\rho}$ consists of all square-integrable, bandlimited signals $x(t)$ of the form 
\[
x(t) = \int_{-W}^{W} X(f)e^{2\pi i ft}df 
\]  
for some $X\in L^2(\RR)$, such that $\norm{x}_{L^2(\RR)}^2 \le E$, and such that moreover the tails of $x(t)$ can be explicitly controlled, in the sense that 
\[
\abs{x(t)} \le \abs{t}^{-\rho}\quad \text{when } \abs{t} \ge T_0 \,.
\]
\end{definition}

The last condition in the definition above might seem somewhat non-standard, but it is in fact not particularly strong: if, for example, the Fourier transform $X$ is absolutely continuous, hence, weakly differentiable, integration by parts readily gives us that 
\begin{align*}
\abs{x(t)}  = \bigg|\frac{1}{2\pi i t}\bigg[&X(W)e^{2\pi i W}-X(-W)e^{-2\pi i W} \nonumber\\ 
&-\int_{-W}^W X'(f)e^{2\pi i ft}df\bigg]\bigg|\nonumber\\
&\leq\frac{ |X(W)|+|X(-W)|+\|X'\|_{L^1(-W,W)}}{2\pi |t|}.
\end{align*}
This is the case, for instance, when $x$ is obtained by band-limiting a time-limited signal, that is,
\[
X(f) = \int_{-T}^{T} x_0(t) e^{-2\pi i ft}dt
\]
for some $x_0\in L^2(\RR)$. The only reason we limit ourselves to the class $\cC_{W,E,T_0,\rho}$, instead of considering simply the class of all band-limited signals with bounded energy, is that in order to use our suggested scheme, we will need access to some samples where we \emph{know} that the signal is sufficiently small. If, instead, we were guaranteed to be given sufficiently many consecutive unfolded samples of $x$, we would have been able to reconstruct the signal without needing to assume any particular tail conditions whatsoever. 

Let $x\in\cC_{W,E,T_0,\rho}$, and suppose that we observe all the samples 
\begin{equation}
x_n := x(n\cdot T_s),\,\,n\in \ZZ
\end{equation}
where $T_s>0$ is the sampling period. The celebrated Shannon-Nyquist sampling theorem guarantees, then, that $x$ could be perfectly reconstructed from the sequence $\left\{x_n\right\}_{n\in \ZZ}$, provided that $\frac{1}{T_s}$ is greater than the so-called Nyquist rate $2W$.

We consider the case where instead of observing $x_n$, we observe a modulo-reduced version of the samples $\left\{x_n^*\right\}_{n\in \ZZ}$, and we would now like to reconstruct $x(\cdot)$. Assuming that $T_s < \frac{1}{2W}$, it would clearly suffice to simply recover the unfolded measurements $\left\{x_n\right\}_{n\in\ZZ}$ from their modulo-reduced versions. 

\subsection{The general recipe}
\label{sec:recipe}

Suppose, for a moment, that $x_0,\ldots, x_{L-1}$ were given to us, and we would now like to recover the proceeding samples $x_L,x_{L+1},\ldots$ from their modulo-reduced version. If we were able to find $L$ numbers $h_1,\ldots,h_{L}$ such that the sequence 
\begin{equation}
e_n := x_n - \left(h_1 x_{n-1} + \ldots + h_L x_{n-L}\right)
\end{equation}
is guaranteed to satisfy that for all $n$, $\abs{e_n} < \frac{\Delta}{2}$, then we could recover $x_L,x_{L+1},\ldots$, in a manner which we now describe. Since $\abs{e_n}<\frac{\Delta}{2}$, clearly, $e^*_n=e_n$. Moreover, since the coefficient of $x_n$ is an integer, we have 
\[
e^*_n = \left[x^*_{n} - \left(h_1 x_{n-1} + \ldots + h_L x_{n-L}\right)\right]^*\,,
\]
hence given $x_{n-1},\ldots,x_{n-L}$ and $x^*_n$, we can compute $e_n$. But then we can also compute $x_n$ as
\[
x_n = e_n + \left( h_1x_{n-1} + \ldots h_L x_{n-L}\right)\,,
\]
and therefore iteratively recover $x_{L},x_{L+1},\ldots $. 

What remains to be done, then, is:
\begin{enumerate}
	\item Find a set of coefficients $h_1,\ldots,h_{L}$, where $L$ may depend on $\Delta$, that guarantees that $|e_n|<\frac{\Delta}{2}$ \emph{universally} for any $x\in\cC_{W,E,T_0,\rho}$.\footnote{In fact, the choice of coefficients $h_1,\ldots,h_{L}$ depends only on $W$ and $E$, whereas the tail behavior dictated by $T_0$ and $\rho$ is of no importance here.} 
	Below, we specify a method based on Chebyshev polynomials for choosing $h_1,\ldots,h_{L}$. This method guarantees that $|e_n|<\frac{\Delta}{2}$ whenever $T_s < \frac{1}{2W}$, whereas the required filter length $L$ grows as $T_s$ gets closer to $\frac{1}{2W}$.
	\item Find a negative integer $N$ such that $x^*_n=x_n$ for all $n< N$. Indeed, by the Riemann-Lebesgue Lemma, $\lim_{\abs{t}\to \infty} x(t) = 0$, so in particular we know that $\abs{x_n} < \frac{\Delta}{2}$ for all small enough $n$. To give a \emph{quantitative} bound for how far we need to go, we use the assumption on the tail  behavior of $x\in \cC_{W,E,T_0,\rho}$: choosing $N< -T_s^{-1}\max(T_0, (\Delta/2)^{-1/\rho}) $ would guarantee us that for all $n < N$, we have $\abs{x_n}<\frac{\Delta}{2}$, hence also that  $x^*_n=x_n$. 
\end{enumerate}

\subsection{Background on Chebyshev polynomials}

The Chebyshev polynomials (of the first kind) are defined by the recursive formula,
\begin{equation}
\begin{split}
T_0(y) &= 1\,, \\
T_1(y) &= y\,, \\
T_K(y) &= 2yT_{K-1}(y) - T_{K-2}(y) \,.
\end{split}
\end{equation}
It is easy to see that for $K\ge 1$, the polynomial $2^{-K+1}T_K(y)$ has degree $K$ and is monic. It is well-known that among all degree-$K$ monic polynomials, it has the smallest maximal value on the interval $[-1,1]$ (for reference, see any textbook on approximation theory, e.g \cite{devore1993constructive}). That maximum is given by
\begin{equation}
\max_{y\in[-1,1]}\abs{2^{-K+1} T_K(y)} = 2^{-K+1} \,.
\end{equation}
On any other interval $[a,b]$, we can easily construct a family of Chebyshev polynomials that take the least maximal value on that interval among all monic degree-$K$ polynomials: the mapping $y\mapsto \frac{2}{b-a}(y-a)-1$ maps $[a,b]$ bijectively into $[-1,1]$, and so 
\begin{equation}
T^{[a,b]}_K(y) = 2\left(\frac{b-a}{4}\right)^K \cdot T_K\left( \frac{2}{b-a}(y-a)-1 \right)
\end{equation}
is clearly the polynomial we need. Its maximal value on $[a,b]$ is given by
\[
\max_{y\in[a,b]}\abs{T_K^{[a,b]}(y)} = 2\left(\frac{b-a}{4}\right)^K \,.
\]
Note that $\max_{y\in[a,b]}\abs{T_K^{[a,b]}(y)} \to 0$ as $K\to\infty$ if and only if $b-a<4$. In the construction that follows, this condition will correspond exactly to the condition that $T_s <\frac{1}{2W}$, that is, we're sampling at a rate strictly above the Nyquist rate.

\subsection{Choosing the Coefficients}
\label{sec:construction}

Denote by $S:\RR^\ZZ \to \RR^\ZZ$ the backward-shift operator, that is,
\[
(Sx)_{n} = x_{n-1} \,.
\]
Writing $x$ in terms of the Fourier inversion formula, we have 
\[
(Sx)_n = \int_{-W}^{W} X(f) e^{2\pi if\cdot (nT_s)} \cdot e^{-2\pi i f T_s} df \,,
\]
hence for any Laurent polynomial $p(z)$,\footnote{That is, an expression of the form $p(z)=a_{N_1}z^{N_1}+a_{N_1-1}z^{N_1-1} + \ldots + a_{-N_2}z^{-N_2}$.} we have 
\[
\left(p(S)x\right)_n =  \int_{-W}^{W} X(f) e^{2\pi if\cdot (nT_s)} \cdot p\left(e^{-2\pi i f T_s}\right) df \,.
\]
Now, choose $h_1,\ldots, h_{2K}$ to be given by the coefficients of the polynomial
\begin{equation}
\begin{split}
p_K(z)
&= z^K\cdot T_K^{[2\cos(2\pi W T_s ),2]}\left(z+z^{-1}\right) \\
&= -h_{2K}z^{2K} - h_{2K-1} z^{2K-1} - \ldots - h_1z + 1 \,.
\end{split} 
\end{equation}
Well,
\begin{align*}
\abs{ e_n } 
&= \abs{ x_{n} - \left(h_1 x_{n-1} + \ldots+ h_{2K}x_{n-2K}\right) } \\
&= \abs{ \left(p_K(S)x\right)_{n} } \\
&= \abs{ \int_{- W}^{ W}X(f) e^{2\pi if\cdot (n T_s)} \cdot p_K\left(e^{-2\pi i f  T_s}\right) df } \\
&\le \norm{X}_{L^1(-W,W)} \cdot \max_{f \in [- W, W]} \abs{p_K\left(e^{-2\pi i f  T_s}\right)}\,.
\end{align*}
We can bound 
\begin{align*}
\norm{X}_{L^1(-W,W)} 
&\le \sqrt{2W} \norm{X}_{L^2(-W,W)} \\ 
&= \sqrt{2W}\norm{x}_{L^2(\RR)} \\
&\le \sqrt{2WE} \,,
\end{align*}
and
\begin{align*}
&\max_{f \in [- W, W]} \abs{p_K\left(e^{-2\pi i f  T_s}\right)} \\
&= \max_{f \in [- W, W]} \abs{T_K^{[2\cos(2\pi W   T_s),2]}(e^{-2\pi i f  T_s} + e^{2\pi i f  T_s}) }\\
&=\max_{f \in [- W, W]} \abs{T_K^{[2\cos(2\pi W   T_s),2]}(2\cos(2\pi f  T_s)) }\\
&= \max_{y \in [2\cos(2\pi W  T_s),2]} \abs{T_K^{[2\cos(2\pi W  T_s),2]}\left(y\right) }\\
&= 2\left(\frac{2-2\cos(2\pi W  T_s)}{4}\right)^K \,.
\end{align*}
Note that whenever $ T_s < \frac{1}{2W}$, we have that $2\cos(2\pi W T_s) > -2$, and so in that case the right-hand-side above tends to $0$ as $K\to \infty$. Thus, we readily obtain

\begin{proposition}
	\label{prop:prop1}
	Choose any 
	\begin{equation}
	K > \frac{ \log\left(\sqrt{32WE}/\Delta\right) }{ \log \left(\frac{2}{1-\cos(2\pi W T_s)}\right)  }
	\end{equation}
	and let $h_1,\ldots,h_{2K}$ be given by the coefficients of the polynomial
	\begin{equation}
	\begin{split}
	p_K(z)
	&= z^K\cdot T_K^{[2\cos(2\pi W  T_s ),2]}\left(z+z^{-1}\right) \\
	&= -h_{2K}z^{2K} - h_{2K-1} z^{2K-1} - \ldots - h_1z + 1 \,.\label{eq:coeffchoice}
	\end{split} 
	\end{equation}
	Then for every $x\in \cC_{W,E,T_0,\rho}$, the sequence
	\[
	e_n 
	=  x_{n} - \left(h_1x_{n-1} + \ldots h_{2K}x_{n-2K}\right) 
	\]
	satisfies that $\abs{e_n} < \frac{\Delta}{2}$ for all $n$.
	
\end{proposition}

We summarize by stating our main result:

\begin{theorem}
Fix $W,E,T_0,\rho\in\RR_+$ and any $\Delta>0$. Then, any $x\in \cC_{W,E,T_0,\rho}$ can be perfectly recovered from $\{x^*(nT_s)\}_{n\in\ZZ}$ provided that $T_s<\frac{1}{2W}$.
\end{theorem}
\begin{proof}
	By Proposition~\ref{prop:prop1}, the recovery procedure we described in Section~\ref{sec:recipe} with $h_1,\ldots,h_{2K}$ chosen as in~\eqref{eq:coeffchoice} correctly recovers the unfolded samples $\left\{x_n\right\}_{n\in \ZZ}$ from the folded samples $\{x^*(nT_s)\}_{n\in\ZZ}$. 
	Now, using the Shannon-Whittaker interpolation formula, see e.g.~\cite{mallat2009wavelet}, we recover $x$ at every point.
\end{proof}


\section{Discussion on  Quantization Noise}

The recovery algorithm described above, relies on filtering $\{x_n\}$ using a monic filter $\{c_n\}$ with vanishing magnitude for the in-band frequencies. By Paley-Wiener's theorem, combined with Jensen's inequality, it is easy to see that any such filter must have unbounded energy in the out-of-band frequencies (i.e., the high frequencies where $\{x_n\}$ has no energy). 

If the input to the recovery algorithm were $\{[x_n+u_n]^*\}$ for some white process $\{u_n\}$, instead of $\{x^*_n\}$, we would therefore have that the result of convolving with $\{c_n\}$ would no longer be restricted to $[-\frac{\Delta}{2},\frac{\Delta}{2})$, due to the contribution of $\{u_n\star c_n\}$ in the out-of-band frequencies.  Modeling the effect of quantization as an additive white noise, we see that the developed algorithm collapses in the presence of quantization noise. This is also the case for the reconstruction algorithm of~\cite{bhandari2017unlimited}, if the order of discrete derivative used there is large enough.

The remedy to this phenomenon is to judiciously balance the in-band energy of $\{c_n\}$  and its out-of-band energy, taking into account the effect of quantization noise, as done in~\cite{othsw18} (see also~\cite{zke08}). However, taking the quantization noise into account yields a lower bound on the in-band energy of $\{c_n\}$, which in turn dictates that $\Delta$ can no longer be arbitrarily small. We conclude that it is essential to take into account quantization  
noise in oversampled systems

\bibliographystyle{IEEEtran}
\bibliography{ref}

\begin{thebibliography}{10}
\providecommand{\url}[1]{#1}
\csname url@samestyle\endcsname
\providecommand{\newblock}{\relax}
\providecommand{\bibinfo}[2]{#2}
\providecommand{\BIBentrySTDinterwordspacing}{\spaceskip=0pt\relax}
\providecommand{\BIBentryALTinterwordstretchfactor}{4}
\providecommand{\BIBentryALTinterwordspacing}{\spaceskip=\fontdimen2\font plus
\BIBentryALTinterwordstretchfactor\fontdimen3\font minus
  \fontdimen4\font\relax}
\providecommand{\BIBforeignlanguage}[2]{{%
\expandafter\ifx\csname l@#1\endcsname\relax
\typeout{** WARNING: IEEEtran.bst: No hyphenation pattern has been}%
\typeout{** loaded for the language `#1'. Using the pattern for}%
\typeout{** the default language instead.}%
\else
\language=\csname l@#1\endcsname
\fi
#2}}
\providecommand{\BIBdecl}{\relax}
\BIBdecl

\bibitem{bhandari2017unlimited}
A.~Bhandari, F.~Krahmer, and R.~Raskar, ``On unlimited sampling,'' in
  \emph{2017 International Conference on Sampling Theory and Applications
  (SampTA)}.\hskip 1em plus 0.5em minus 0.4em\relax IEEE, 2017, pp. 31--35.

\bibitem{bkr18}
A.~{Bhandari}, F.~{Krahmer}, and R.~{Raskar}, ``Unlimited sampling of sparse
  sinusoidal mixtures,'' in \emph{2018 IEEE International Symposium on
  Information Theory (ISIT)}, June 2018.

\bibitem{ct17}
M.~Cucuringu and H.~Tyagi, ``On denoising modulo 1 samples of a function,''
  \emph{arXiv preprint arXiv:1710.10210}, 2017.

\bibitem{mjg18}
O.~Musa, P.~Jung, and N.~Goertz, ``Generalized approximate message passing for
  unlimited sampling of sparse signals,'' \emph{arXiv preprint
  arXiv:1807.03182}, 2018.

\bibitem{rass18}
S.~{Rudresh}, A.~{Adiga}, B.~A. {Shenoy}, and C.~S. {Seelamantula},
  ``Wavelet-based reconstruction for unlimited sampling,'' in \emph{2018 IEEE
  International Conference on Acoustics, Speech and Signal Processing
  (ICASSP)}, April 2018.

\bibitem{rbwp18}
L.~Rencker, F.~Bach, W.~Wang, and M.~D. Plumbley, ``Sparse recovery and
  dictionary learning from nonlinear compressive measurements,'' \emph{arXiv
  preprint arXiv:1809.09639}, 2018.

\bibitem{sh18}
V.~Shah and C.~Hegde, ``Signal reconstruction from modulo observations,''
  \emph{arXiv preprint arXiv:1812.00557}, 2018.

\bibitem{jt19}
F.~Ji, W.~P. Tay \emph{et~al.}, ``Recovering graph signals from folded
  samples,'' \emph{arXiv preprint arXiv:1903.03741}, 2019.

\bibitem{othsw18}
O.~{Ordentlich}, G.~{Tabak}, P.~K. {Hanumolu}, A.~C. {Singer}, and G.~W.
  {Wornell}, ``A modulo-based architecture for analog-to-digital conversion,''
  \emph{IEEE Journal of Selected Topics in Signal Processing}, vol.~12, no.~5,
  pp. 825--840, Oct 2018.

\bibitem{boufounos12}
P.~T. Boufounos, ``Universal rate-efficient scalar quantization,'' \emph{IEEE
  Transactions on Information Theory}, vol.~58, no.~3, pp. 1861--1872, March
  2012.

\bibitem{oe13b}
O.~Ordentlich and U.~Erez, ``Integer-forcing source coding,'' \emph{IEEE
  Transactions on Information Theory}, vol.~63, no.~2, pp. 1253--1269, Feb
  2017.

\bibitem{ro19}
E.~Romanov and O.~Ordentlich, ``Blind unwrapping of modulo reduced {G}aussian
  vectors: Recovering {MSB}s from {LSB}s,'' \emph{arXiv preprint
  arXiv:1901.10396}, 2019.

\bibitem{gershogray}
A.~Gersho and R.~M. Gray, \emph{Vector quantization and signal
  compression}.\hskip 1em plus 0.5em minus 0.4em\relax Springer Science \&
  Business Media, 2012, vol. 159.

\bibitem{devore1993constructive}
R.~A. DeVore and G.~G. Lorentz, \emph{Constructive approximation}.\hskip 1em
  plus 0.5em minus 0.4em\relax Springer Science \& Business Media, 1993, vol.
  303.

\bibitem{mallat2009wavelet}
S.~Mallat \emph{et~al.}, ``A wavelet tour of signal processing: The sparse
  way,'' \emph{[3rd ed.] ed. Amsterdam: Elsevier Academic Press}, 2009.

\bibitem{zke08}
R.~Zamir, Y.~Kochman, and U.~Erez, ``Achieving the {G}aussian rate-distortion
  function by prediction,'' \emph{IEEE Transactions on Information Theory},
  vol.~54, no.~7, pp. 3354--3364, July 2008.

\end{thebibliography}

\end{document}